\documentclass[a4paper,11pt]{article}
\usepackage{a4wide}
\usepackage[colorlinks=true, linkcolor=blue, citecolor=blue, urlcolor=blue]{hyperref}
\usepackage{my_package}
\newcommand{\prox}{{\textbf{(Cont)}}\xspace}
\title{Asymptotic stability and ergodic properties of quantum trajectories under imperfect measurement}
\author{
    Nina H. Amini\thanks{Université Paris-Saclay, CNRS, CentraleSupélec, Laboratoire des signaux et systèmes (L2S), France}
    \and Tristan Benoist\thanks{Univ Toulouse, INUC, UT2J, INSA Toulouse, TSE, CNRS, IMT, Toulouse, France}
    \and Ma\"el
 Bompais\thanks{School of Mathematical Sciences, University of Nottingham, United Kingdom}
    \and Cl\'ement
 Pellegrini\footnotemark[2]
}

\makeatletter
\renewcommand{\@date}{}
\makeatother

\begin{document}

\maketitle

\begin{abstract}
   We investigate the asymptotic stability and ergodic properties of quantum trajectories under imperfect measurement, extending previous results established for the ideal case of perfect measurement. We establish a necessary and sufficient condition ensuring the convergence of the estimated trajectory, initialized from an estimated state, to the true trajectory. This result is obtained assuming that the associated quantum channel is irreducible. Building on this, we prove the uniqueness of the invariant measure and demonstrate  convergence toward this measure.
\end{abstract}

\tableofcontents
\section{Introduction}
Quantum trajectories are random processes describing the evolution of quantum systems subject to repeated indirect measurements. An indirect measurement procedure is typically implemented by letting probes successively interact with the system of interest, and (directly) measuring them afterwards. Because the interaction between the system and the probe creates correlations (e.g. entanglement), the measurement outcome depends on the system’s state and induces a back action on it. The probes are supposed to be independent, initialized in the same state, and the resulting evolution for the state of the system is then a Markov chain on the set of density matrices.

When the probes are all prepared in a pure state and the observable measured on them has non-degenerate spectrum, we say that the measurement is \textit{perfect}. However, in experimental situations, imperfections may arise. For instance, the probes may be prepared in a mixed state, the measurement may be degenerate, the detector may be biased and provide incorrect outcomes part of the time, or the system may interact with an unmonitored environment. All these imperfections can be incorporated into the same formalism: instead of updating the state through a single Kraus operator, we use Kraus maps, i.e., completely positive trace non-increasing maps, which can be represented as weighted sums of actions by Kraus operators.

In this paper, our aim is to study the asymptotic behavior of quantum trajectories evolving with this \textit{imperfect} dynamics. In the literature, there exist many developments concerning quantum trajectories under perfect measurements, some of which we briefly review in what follows. Under a condition called \textit{purification} derived in \cite{maassen2006purification}, the trajectory almost surely purifies over time, which means that, starting from a potentially mixed state, it asymptotically reaches the set of pure states. This condition also ensures what is called the \textit{asymptotic stability} of the quantum trajectory (with respect to its initial state): when purification occurs, a wrongly initialized quantum trajectory evolving with the measurement results detected on the true system will converge to the true trajectory, see \cite{benoist2019invariant} for discrete-time measurement and \cite{benoist2021invariant,amini2021asymptotic} for continuous-time measurement. 
This asymptotic stability was used in \cite{benoist2019invariant,benoist2021invariant} to show that quantum trajectories arising from perfect and irreducible indirect measurements admit a unique invariant measure, when viewed as Markov processes.
 This result was later refined by demonstrating several limit theorems \cite{benoist2023limit,benoist2025quantum_a}, including an ergodic theorem (also referred to as Law of Large Numbers) for arbitrary continuous functions of the state, as well as a Central Limit Theorem, and a restricted Large Deviation Principle. 


By contrast, fewer results are known for quantum trajectories under imperfect measurements. Among the first important and pioneer results, an ergodic theorem was established in \cite[Section~6]{kummerer2004pathwise} for linear functions of the state.\footnote{While proved only for perfect measurements in that reference, the proof extends to imperfect measurements readily.} For quantum non-demolition (QND) imperfect quantum trajectories, one can mention the convergence to a pointer state established in \cite{bompais2023asymptotic}, as well as the asymptotic stability, i.e., the convergence of the estimated trajectory toward the same pointer state.
 Concerning the stability for imperfect measurements without the QND assumption, the fidelity between the true and the estimated trajectory is known to be a sub-martingale \cite{rouchon2011fidelity,somaraju2012design,amini2014stability}.
The question of stability in the filtering formalism has also been  studied, both in the classical setting \cite{ocone1996asymptotic,delyon1988lyapunov,van2009observability} and in its quantum counterpart \cite{handel2009stability}.

In this article, we prove in Theorem~\ref{thm:asymp_stab_imperfect} the asymptotic stability of irreducible imperfect measurement dynamics under a contractivity assumption, and we show that this assumption is a necessary and sufficient condition in the case of irreducible quantum channels. This assumption requires the existence of a sequence of Kraus maps that accumulates to a rank-one map. When restricted to perfect measurements (a special case of our setting), it is equivalent to the purification assumption (see Proposition~\ref{prop:pur_implies_cont}).
As in \cite{benoist2019invariant}, leveraging asymptotic stability, we establish the uniqueness of the invariant measure for imperfect quantum trajectories. We also show that quantum trajectories converge in distribution toward this invariant measure, as summarized in Theorem~\ref{thm:invariant_measure_convergence}. A consequence of this theorem is an ergodic theorem for continuous functions, stated as Theorem~\ref{thm:convergence_of_the_ergodic_mean}. Its proof is a direct adaptation of the arguments developed in \cite{benoist2023limit}.
Note that, compared to \cite{benoist2019invariant}, we do not obtain a convergence rate, since we do not derive a rate for asymptotic stability. Such a rate would be necessary to establish further limit theorems as in \cite{benoist2023limit, benoist2025quantum_a}.

The paper is organized as follows. Section~\ref{sec:qtraj_def} introduces the model of quantum trajectories under imperfect measurements. We set up an appropriate probability space on which these trajectories are defined and formulate the contractivity and irreducibility assumptions. Section~\ref{sec:asymptotic_stability} is devoted to the analysis of asymptotic stability. Section~\ref{sec:uniqueness} then presents our main results on the invariant measure and on the convergence in distribution of quantum trajectories toward this measure.

\section{Imperfect quantum trajectories: model and assumptions}\label{sec:qtraj_def}
In this section, we present the model of quantum trajectories that incorporates measurement imperfections. We then introduce an appropriate probability space on which these trajectories are defined. Finally, we state the two key assumptions, irreducibility and contractivity, that will be used throughout the rest of the article.

\subsection{Dynamics under imperfect measurements}

We consider a quantum system represented by a complex Hilbert space $\H$ of dimension $d<\infty$. The state space of the system is the set of density operators,
$$\S(\H)=\{ \rho \in \B(\H) ~|~\rho \geq 0,~\tr{\rho}=1 \}.$$
Here $\B(\H)$ denotes the set of linear operators on $\H$ and the notion of positivity is the semi-definite one. We equip $\S(\H)$ and its linear span with the trace norm denoted $\|\cdot\|$ and $\B(\H)$ with the operator norm denoted $\|\cdot\|_\infty$.

The system is assumed to undergo a sequence of independent indirect measurements. In this setting, a single measurement step is described by a quantum instrument $\{\Phi_i\}_{i=1}^m$, that is, a family of completely positive maps
\[
\Phi_i : \mathcal{B}(\mathcal{H}) \to \mathcal{B}(\mathcal{H}), \qquad i = 1,\dots,m,
\]
whose sum $\Phi := \sum_{i=1}^m \Phi_i$ is trace-preserving and is called a quantum channel.
Given an input state $\rho$, the probability of observing outcome $i$ is
\[
p(i|\rho) = \operatorname{tr}\!\left[\Phi_i(\rho)\right],
\]
and, conditional on this outcome, the (unnormalized) post-measurement state is $\Phi_i(\rho)$, normalized as $\Phi_i(\rho)/p(i|\rho)$ whenever $p(i|\rho) > 0$.

Iterating this mechanism generates a stochastic evolution $(\rho_n)_{n\geq 0}$, referred to as a quantum trajectory. Starting from an initial state $\rho_0$, the system evolves in discrete time. 
At each step, corresponding to an indirect measurement, the state is updated according to
\[
\rho_{n+1} = \frac{\Phi_{i}(\rho_n)}{\operatorname{tr}(\Phi_{i}(\rho_n))} 
\qquad \text{with probability } \tr{\Phi_i(\rho_n)}.
\]In particular, if $(i_1,\ldots,i_{n})$ have been observed during the first $n$ measurement steps, the state of the system at time $n$ conditioned on this sequence of outcomes is 
\[
\rho_n = \normalized{\Phi_{i_{n}} \circ \cdots \circ \Phi_{i_1} (\rho_0)}.
\]
Note that $\tr{\Phi_{i_{n}} \circ \cdots \circ \Phi_{i_1} (\rho_0)}>0$ almost surely, so the process is well defined.

The corresponding averaged evolution is obtained by averaging over the outcomes, which yields the deterministic evolution via the quantum channel:
\[
\mathbb{E}[\rho_{n+1} \mid \rho_n] = \sum_{i=1}^m \Phi_i(\rho_{n}) = \Phi(\rho_n).
\]
Iterating, it follows that
\[ \E{\rho_n}  =  \Phi^n(\rho_0). \]

In the perfect detection regime, each outcome $i$ is associated with a single Kraus operator $V_i$, so that the corresponding Kraus maps read $\Phi_i(\rho) = V_i \rho V_i^\ast.$

In a realistic setup, the detector may be biased, so the reported outcome $i$ does not necessarily coincide with the true microscopic event $j$. This bias is described by a stochastic matrix
\[
\eta = (\eta_{i,j})_{i,j}, \qquad \eta_{i,j} \geq 0\ \forall i,j,\quad \sum_i \eta_{i,j} = 1\  \forall j,
\]
where $\eta_{i,j}$ denotes the probability that the detector reports $i$ when an ideal detector would have reported $j$.
As a result, the effective measurement dynamics is no longer governed by the Kraus operators $\{V_i\}_i$, but by an instrument $\{\Phi_i\}_i$ defined as (see \cite{somaraju2012design}):
\[
\Phi_i(\rho) = \sum_{j} \eta_{i,j}V_j \rho V_j^\ast.
\]
This type of detection imperfection is routinely encountered in realistic experimental implementations; see, for example, \cite{sayrin2011real,somaraju2012design}. The framework of quantum instruments is, however, much more general and can accommodate a wide range of imperfections, including those mentioned in the introduction.

\subsection{Probability space}
We now introduce the underlying probability space, which will serve as the rigorous framework for the subsequent analysis.

Let $\N$ denote the set of natural integers $\{1,2,3,\ldots   \}.$ Let $\O=\{1,\ldots,m\}$ be the finite set of possible measurement outcomes. We set 
$$\Omega=\O^{\N}$$ 
i.e., $\Omega$ is the set of all possible infinite sequences $\omega=(i_1,i_2,i_3,\ldots)$ of measurement results. For $k \in \N$ and $(i_1,\ldots,i_{k}) \in \O^k$, we define the cylinder subset
$$C_{i_1,\ldots,i_{k}}=\{\omega \in \Omega \mid \omega_1=i_1, \ldots, \omega_{k}=i_{k} \}.$$
For a fixed $n \in \N,$ we consider the $\sigma$-algebra $\F_n$ generated by all the cylinders of size less than or equal to $n,$
$$\F_n = \sigma \oo C_{i_1,\ldots,i_{k}},~ k \leq n,~ (i_1,...,i_k)\in \O^k     \cc.$$
Then equip $\Omega$ with the $\sigma$-algebra $\F$ generated by all the $\F_n,$ $n \in \N,$
$$ \F = \underset{n \in \N}{\bigvee}  \F_n.$$
Given a state $\rho_0 \in \S(\H),$ we define the probability measure $\P^{\rho_0}{\big|_{\F_n}}$ on $(\Omega,\F_n),$ governing the trajectory up to time $n,$ by assigning the following probability to each cylinder $C_{i_1,\ldots,i_k},$ $k \leq n,$
$$\P^{\rho_0}{\big|_{\F_n}} \left(C_{i_1,\ldots,i_k} \right) = \tr{ \Phi_{i_{k}}\circ\ldots\circ\Phi_{i_1} (\rho_0) }.$$
Since $\Phi=\sum_{i=1}^m\Phi_i$ is trace-preserving and the trace is linear, the probability measures $\P^{\rho_0}{\big|_{\F_n}}$, $n\in \N,$ form a consistent family. By Kolmogorov's extension theorem, these finite-time measures extend uniquely to a probability measure $\P^{\rho_0}$ on $(\Omega,\F)$.

The quantum trajectory $(\rho_n)$ is itself a sequence of $(\F_n)$-adapted random variables, satisfying
$$ \rho_n(\omega) = \normalized{\Phi_{\omega_n} \circ \ldots \circ \Phi_{\omega_1} (\rho_0) }.$$
With respect to $\P^{\rho_0}$, $(\rho_n)$ has the law of the quantum trajectory defined by the instrument $\{\Phi_i\}_{i=1}^m$ and initial state $\rho_0\in \S(\H)$.
\subsection{Irreducibility and contractivity assumptions}
In this paper, we introduce the following two assumptions. Together, they allow us to prove the asymptotic stability of quantum trajectories. Using this result, we then establish their ergodic properties.
\begin{assumption}[\textbf{Irr}]
The quantum channel $\Phi$ is irreducible, meaning that the only non-zero orthogonal projection $P$ on $\H$ satisfying
$\Phi\bigl(P \B(\H) P\bigr) \subseteq P \B(\H) P$ is $P = \mathrm{Id}_\H.$
\end{assumption}

Equivalently -- see \cite[Theorem~3.1]{evans1978spectral} -- $\Phi$ admits a unique invariant state $\rho_{\inv}$, and this state is full-rank, i.e.
\[
\Phi(\rho_{\inv})=\rho_{\inv}
\quad\text{and}\quad
\rho_{\inv}>0.
\]
This assumption is the quantum analogue of irreducibility for classical finite state Markov chains, for which it ensures uniqueness of the invariant measure. 
It is a standard structural condition in the study of the asymptotic behavior of quantum trajectories, see e.g. \cite{kummerer2004pathwise, benoist2019invariant}.\\

Let us now state the contractivity assumption. For convenience, we introduce the notation
$$\Phi_I := \Phi_{i_n} \circ \cdots \circ \Phi_{i_1}$$
for any word $I = (i_1,\ldots,i_n) \in \O^n,~n \in \N.$

\begin{assumption}[\textbf{Cont}]
There exists a sequence $(I_n)_{n\in \N}\subset \cup_{k\in \N}\O^k$ and positive semi-definite $X, Z\in \B(\H)$ such that
$$ \frac{\Phi_{I_{n}}}{\norm{ \Phi_{I_{n}} }} \xrightarrow[n \to \infty]{}Z\tr{X~.~},$$
where $\tr{X~\cdot~} Z$ denotes the map $\rho \mapsto Z\tr{X \rho}.$
\end{assumption}
Here, the norm used is the operator norm related to the trace norm over $\S(\H)$. It is defined by $\|\Psi\| = \sup_{\rho\in \S(\H)} \|\Psi(\rho)\|$ for any positive map $\Psi$. Since $\B(\H)$ is finite-dimensional, this choice is not crucial and any other norm could be used.

This assumption is the analogue of contractivity in the classical theory of products of i.i.d. random matrices
\cite{furstenberg1963noncommuting, guivarch1980proprietes}, where a semigroup of matrices is called contractive
if it contains an element with a simple dominant singular value.
Equivalently, there exists a sequence of matrices whose normalized product converges to a rank-one operator, corresponding to a projection onto the image space associated with the dominant singular value.
Combined with strong irreducibility, this condition yields the existence of a unique stationary measure on the projective space; see \cite{bougerol} for an introduction to this topic.

\medskip

We now present two situations where \cont is satisfied.
\begin{proposition}\label{thm:PI_implies_Cont}
Assume there exists a word $I$ such that $T:=\Phi_I$ is primitive, i.e. there exists $n \in \N$ such that for all $X\geq 0$ with $X\neq 0$, one has $T^n(X)> 0$. Then \cont holds.
\end{proposition}
Note that $\Phi_I$ being primitive implies that \irr holds -- see \cite[Proposition~2.2]{evans1978spectral}.
\begin{proof}
By Perron--Frobenius theory for primitive completely positive maps -- see \cite[Theorems~2.3~and~2.4]{evans1978spectral} --, the spectral radius $r$ of the trace dual $T^*$ of $T$ is such that there exists $X>0$ such that $T^*(X)=rX$. Then, $\widehat{T}:\rho\mapsto r^{-1}X^{\frac12}T(X^{-\frac12} \rho X^{-\frac12})X^{\frac12}$ is trace-preserving. Since $T$ is primitive, so is $\widehat{T}$ and \cite[Theorem~4.2]{evans1978spectral} implies 
$$\widehat{T}^n\xrightarrow[n\to\infty]{}\sigma\tr{~\cdot~}$$
for some positive definite $\sigma\in\S(\H)$.
It follows that 
$$r^{-n}T^n \xrightarrow[n\to\infty]{}X^{-\frac12}\sigma X^{-\frac12}\tr{X~\cdot~}.$$
Hence,
$$
\frac{T^n}{\|T^n\|}\ \xrightarrow[n \to \infty]{}
 \frac{P}{\|P\|}, \quad\textrm{with},
\quad
P=X^{-\frac12}\sigma X^{-\frac12}\tr{X~\cdot~}.
$$
Setting 
\(
Z := X^{-\frac12}\sigma X^{-\frac12},
\)
we obtain
$$
\frac{T^k}{\|T^k\|}\ \xrightarrow[k \to \infty]{} Z\tr{X~\cdot~}.
$$
For $n\ge 1$, let $I_{n}$ denote the $n$–fold concatenation of $I$. 
Then $\Phi_{I_{n}} = T^{n}$ and
$$
\frac{\Phi_{I_{n}}}{\| \Phi_{I_{n}} \|}
\ \xrightarrow[n \to \infty]{} 
Z\tr{X ~\cdot~},
$$
which is precisely Assumption \cont.
\end{proof}

Recall that in the perfect measurement case each map $\Phi_i$ has Kraus rank $1$. Namely, there exist matrices $V_i\in \B(\H)$ such that
\[
\Phi_i(\rho)=V_i\,\rho\,V_i^\dagger,
\qquad \sum_i V_i^\dagger V_i=\id_\H .
\]

We will now give a sufficient condition to verify \cont when some of the maps $\Phi_i$ are in this perfect form.

For a collection of operators $\{V_i\in \B(\H)\}$, we say that it satisfies the non-darkness condition, denoted \textbf{(ND)}, if, for every subspace $E \subset \H$ with $\dim E \ge 2$, there exists a finite sequence $(i_1,\dots,i_n)$ such that
$$
\pi_E V_{i_n}^\dagger \cdots V_{i_1}^\dagger V_{i_1} \cdots V_{i_n} \pi_E
\neq \|V_{i_n} \cdots V_{i_1} \pi_E\|_\infty^2 \pi_E,
$$
where $\pi_E$ is the orthogonal projector onto $E$.\\

\begin{proposition}\label{prop:pur_implies_cont}
Let $\Phi = p \Phi_1 + (1-p) \Phi_2$ with $p \in (0,1],$ where $\Phi_1$ and $\Phi_2$ are quantum channels.
Consider a perfect unraveling $\{\Phi_{1,i}:\rho\mapsto V_i\rho V_i^*\}$ of $\Phi_1$ such that $\{V_i\}$ satisfies \nondarkness. Consider the quantum instrument defined by $\{p\Phi_{1,i}\}\cup\{(1-p)\Phi_{2,i}\}$ for some instrument $\{\Phi_{2,i}\}$ summing up to the quantum channel $\Phi_2$. Then \cont\ holds for this instrument.
\end{proposition}
Remark that the instrument $\{p\Phi_{1,i}\}\cup\{(1-p)\Phi_{2,i}\}$ corresponds to the experiment in which, at each step, one chooses independently to use the instrument $\{\Phi_{i,1}\}$ with probability $p$ and the instrument $\{\Phi_{2,i}\}$ with probability $1-p$. When $\Phi_1=\Phi_2=\Phi$ and the second instrument is $\{\Phi\}$ that corresponds to missing a measurement outcome with probability $1-p$ at each time step. Hence, for $p=1$ one recovers the setting of \cite{benoist2019invariant}.
\begin{proof}
    Consider the quantum channel $\Phi_1$ with perfect measurement and Kraus operators $\{V_i\}$. Since \textbf{(ND)} is satisfied, this quantum channel satisfies the setting of \cite{benoist2019invariant}. In particular \cite[Proposition~2.2]{benoist2019invariant} yields that almost surely, there exist $x\in \H$ such that $\|x\|=1$ and
 \[
 \frac{V_{i_1}^*\dotsb V_{i_k}^*V_{i_k} \cdots V_{i_1}}{\|V_{i_k} \cdots V_{i_1}\|^2_\infty}\xrightarrow[k\to\infty]{}|x\rangle\langle x|.
 \]
 Since $\frac{V_{i_k} \cdots V_{i_1}}{\|V_{i_k} \cdots V_{i_1}\|_\infty}$ has supremum norm $1$, it takes value in a compact set. Hence, there exists a sequence $(I_n)_n\subset \cup_{k\in\N}\O^k$ and $z\in\H$, $\|z\|=1$ such that denoting $I_n=(i_{n,1},\dotsc,i_{n,k_n})$,
 $$\frac{V_{i_{n,k_n}} \cdots V_{i_{n,1}}}{\|V_{i_{n,k_n}} \cdots V_{i_{n,1}}\|_\infty}\xrightarrow[n\to\infty]{}|z\rangle\langle x|.$$
 Considering $\Phi_{I_n}=\Phi_{1,i_{n,k_n}}\circ\ldots\circ\Phi_{1,i_{n,1}}$,
 $$ \frac{\Phi_{I_{n}}}{\norm{ \Phi_{I_{n}} }} \xrightarrow[n \to \infty]{}Z\tr{X~.~},$$
 where $X= \vert x \rangle\langle x \vert$ and $Z=\vert z\rangle\langle z\vert.$
\end{proof}

\section{Asymptotic stability}\label{sec:asymptotic_stability}

In this section, we consider two quantum trajectories evolving according to the same sequence of measurement outcomes: 
the first one starts from the true initial state $\rho_0$, while the second one starts from a different initial state, say $\hat\rho_0\in \S(\H)$. 
The measurement outcomes are generated by the true trajectory, and the question is whether the two trajectories asymptotically coincide. 
More precisely, if $(\rho_n)$ and $(\hat\rho_n)$ denote the two trajectories, then conditionally on observing $(i_1,\ldots,i_n)$,

$$\hat\rho_n = \normalized{\Phi_{i_n}\circ \ldots \circ \Phi_{i_1}(\hat\rho_0)}.$$

Considering that the law of the outcomes $(i_n)_n\in \O^\N$ is given by $\P^{\rho_0}$, the process $(\hat\rho_n)$ is called the estimated trajectory. Note that, in contrast to $(\rho_n)$, the process $(\hat\rho_n)$ is not Markov. For the well-posedness of the estimated trajectory—namely, the non-vanishing of its denominator—it is standard to assume that the initial estimated state satisfies  
\begin{equation}\label{eq:state order}
\ker(\hat\rho_0)\subset \ker(\rho_0).
\end{equation}
Indeed, when this relation holds, there exists a constant $c>0$ such that $\rho_0 \leq c\,\hat\rho_0$. Hence, for all $(i_1,\ldots,i_n)\in\mathcal O^n$,
$$
\tr{\Phi_{i_n}\circ\cdots\circ\Phi_{i_1}(\rho_0)}
   \le c\;\tr{\Phi_{i_n}\circ\cdots\circ\Phi_{i_1}(\hat\rho_0)},
$$
which implies
$$
\tr{\Phi_{i_n}\circ\cdots\circ\Phi_{i_1}(\hat\rho_0)}=0
\;\Longrightarrow\;
\P_{\rho_0}(i_1,\ldots,i_n)=0.
$$
Then the denominator of $\hat\rho_n$ is $\P_{\rho_0}$-almost surely non-vanishing. When Equation~\ref{eq:state order} holds, we write $\rho_0\ll\hat\rho_0$ and we adopt this assumption throughout the paper.

The notion of asymptotic stability refers to the convergence of $(\hat\rho_n)$ toward $(\rho_n).$ To quantify the closeness between the states $\rho_n$ and $\hat\rho_n$, we use the \emph{quantum fidelity} \cite{jozsa1994fidelity}, defined for any density matrices $\rho$ and $\sigma$ by
\[
F(\rho,\sigma) = \tr*[][2]{\sqrt{\sqrt{\rho}\sigma \sqrt{\rho}}}.
\]
The fidelity takes values in $[0,1]$ and satisfies $F(\rho,\sigma)=1$ if and only if $\rho=\sigma$, and $F(\rho,\sigma)=0$ if and only if $\range(\rho) \perp \range(\sigma).$\\

Among the fundamental properties of the fidelity, we have the following result from \cite{rouchon2011fidelity}. Note that it holds in the general case, even if $\irr$ and $\prox$ do not hold.

\begin{theorem}[Theorem~2 in \cite{rouchon2011fidelity}]\label{thm:fidelite_sous_mart}
    The fidelity between $(\rho_n)$ and $(\hat\rho_n)$ is a submartingale, i.e. for all $n \geq 0$,
    $$\EE{F(\rho_{n+1}, \hat\rho_{n+1})}{\F_n} \geq F(\rho_n, \hat\rho_n)\qquad \P^{\rho_0}-\mathrm{a.s.}$$
    In particular, the bounded sequence $(F(\rho_n,\hat\rho_n))$ converges almost surely.
\end{theorem}

While the submartingale property guarantees almost sure convergence of the fidelity, 
it provides no information on the value of the limit. The following theorem shows that, in the irreducible case, the contractivity assumption is a necessary and sufficient condition for the fidelity to converge to one.

\begin{theorem}[Asymptotic stability]\label{thm:asymp_stab_imperfect} 
Assume that $\irr$ holds. Then the following equivalence holds:
$$
\cont \quad \Longleftrightarrow \quad 
\forall\, \rho_0,\hat\rho_0 \text{ with } \rho_0 \ll \hat\rho_0,\ 
\limn F(\rho_n,\hat\rho_n)=1 \quad \P^{\rho_0}-\mathrm{a.s.}
$$
\end{theorem}

Before proving this theorem, we require the following lemma.
\begin{lemma}\label{lem:acc point pos}
    For any $\rho\in \S(\H)$, $\P^\rho$-almost surely, any accumulation point $\Lambda$ of $\left(\frac{\Phi_{\omega_n}\circ\dotsb\circ\Phi_{\omega_1}}{\|\Phi_{\omega_n}\circ\dotsb\circ\Phi_{\omega_1}\|}\right)_{n\in\N}$ is such that $\Lambda(\rho)\neq0$.
\end{lemma}
\begin{proof}
    Let $\sigma=\id_\H/d$ and 
    $$M_n=\frac{\Phi_{\omega_1}^*\circ\dotsb\circ\Phi_{\omega_n}^*(\id_\H)}{\P^\sigma(\omega_1,\dotsc,\omega_n)}.$$
    Then, by construction $\|M_n\|_\infty\leq d$ and $\Phi^*(\id_\H)=\id_\H$ implies
    $$\EE[][\sigma]{M_{n+1}}{\F_n}=M_n.$$
    Thus, $(M_n)_n$ is a uniformly bounded $\P^\sigma$ martingale. It follows it converges almost surely and in $L^1(\P^\sigma)$-norm to a random variable $M_\infty$.

    Moreover, by definition, for any $\rho\in \S(\H)$, $\d\P^\rho|_{\F_n}=\tr{\rho M_n}\d\P^\sigma|_{\F_n}$. Since $(M_n)_n$ converges in $L^1(\P^\sigma)$-norm to $M_\infty$,
    \begin{equation}
        \label{eq:radon-nikodym}
        \d\P^\rho=\tr{\rho M_\infty}\d\P^\sigma.
    \end{equation}
    In particular, $\P^\rho$ is absolutely continuous with respect to $\P^\sigma$.

    Let $\Lambda$ be an accumulation point of $\left(\frac{\Phi_{\omega_n}\circ\dotsb\circ\Phi_{\omega_1}}{\|\Phi_{\omega_n}\circ\dotsb\circ\Phi_{\omega_1}\|}\right)_{n\in\N}$. Then, since $\tr{\sigma M_\infty}=1$ by definition, $\|M_\infty\|_\infty>0$ $\P^\rho$-almost surely and $\Lambda^*(\id_\H)=\frac{M_\infty}{\|M_\infty\|_\infty}$. It follows that $\tr{\Lambda(\rho)}=\tr{\rho M_\infty}/\|M_\infty\|_\infty$. Then Eq.~\eqref{eq:radon-nikodym} implies $\tr{\rho M_\infty}>0$ $\P^\rho$-almost surely, so $\Lambda(\rho)\neq 0$ $\P^\rho$-almost surely.
\end{proof}

\begin{proof}[Proof of Theorem~\ref{thm:asymp_stab_imperfect}]
We first prove the forward implication. Assume that \cont holds. We leverage the almost sure convergence of the fidelity established in Theorem~\ref{thm:fidelite_sous_mart}. Since $F(\rho_n,\hat \rho_n)$ is bounded by one, Lebesgue's dominated convergence theorem implies that
$$\E*[][\rho_0]{\sum_{p=1}^\infty \tfrac1{p^2}|F(\rho_{n+p},\hat \rho_{n+p})-F(\rho_n,\hat \rho_{n})|}\xrightarrow[n\to\infty]{}0.$$
Hence,
$$\E*[][\rho_0]{\sum_{p=1}^\infty \tfrac1{p^2}\EE{|F(\rho_{n+p},\hat \rho_{n+p})-F(\rho_n,\hat \rho_{n})|}{\F_n}}\xrightarrow[n\to\infty]{}0.$$
That implies the convergence in $L^1$-norm of
$$\sum_{p=1}^\infty \tfrac1{p^2}\EE{|F(\rho_{n+p},\hat \rho_{n+p})-F(\rho_n,\hat \rho_{n})|}{\F_n}$$
to $0$.
Thus, there exists a subsequence $(n_k)_k$ such that 
$$\sum_{p=1}^\infty \tfrac1{p^2}\EE{|F(\rho_{n_k+p},\hat \rho_{n_k+p})-F(\rho_{n_k},\hat \rho_{n_k})|}{\F_{n_k}}\xrightarrow[k\to\infty]{\P^{\rho_0}\text{-a.s.}}0.$$
Then, for any $p\in \N$,
$$\EE{|F(\rho_{n_k+p},\hat \rho_{n_k+p})-F(\rho_{n_k},\hat \rho_{n_k})|}{\F_{n_k}}\xrightarrow[k\to\infty]{}0$$
$\P^{\rho_0}$-almost surely, which can be rewritten as
\begin{equation*}
    \lim_k \sum_{I \in \O^p} \abs*{F \left( \normalized{\Phi_{I} (\rho_{n_k})}  , \normalized{\Phi_{I} (\hat\rho_{n_k})} \right)  - F \left( \rho_{n_k} , \hat\rho_{n_k} \right) }
     \tr{\Phi_{I} ( \rho_{n_k} )} = 0.
\end{equation*}
Since all the terms in the sum are non-negative, it follows that, $\P^{\rho_0}$-almost surely, for any fixed word $I \in \O^p$, 
\begin{equation}\label{eq:convergence words}
    \lim_k \abs*{F \left( \normalized{\Phi_{I} (\rho_{n_k})}  , \normalized{\Phi_{I} (\hat\rho_{n_k})} \right)  - F \left( \rho_{n_k} , \hat\rho_{n_k} \right) }
     \tr{\Phi_{I} ( \rho_{n_k} )} = 0.
\end{equation}
Let 
$$\Lambda_n=\frac{\Phi_{\omega_n}\circ\dotsb\circ\Phi_{\omega_1}}{\|\Phi_{\omega_n}\circ\dotsb\circ\Phi_{\omega_1}\|}.$$
Since $\norm{\Phi_{\omega_{n}} \circ \ldots \circ \Phi_{\omega_1}} = 0$ implies 
$\P^{\rho_0}(\omega_1,\dotsc,\omega_n)=\tr{\Phi_{\omega_{n}} \circ \ldots \circ \Phi_{\omega_1}(\rho_0)} = 0,$ the map $\Lambda_n$ is $\P^{\rho_0}$-almost surely well defined.
For all $n \geq 0$, $\rho_n$ and $\hat\rho_n$ satisfy 
$$ \rho_n = \normalized{\Lambda_n(\rho_0)} \qquad \text{and} \qquad \hat\rho_n = \normalized{\Lambda_n(\hat\rho_0)},$$
where these expressions are well-defined since $\Lambda_n(\rho_0) \neq 0$ and $\Lambda_n(\hat\rho_0) \neq 0$ $\P^{\rho_0}$--a.s.
The sequence $(\Lambda_n)$ takes value in the compact set of completely positive maps from $\B(\H)$ to itself with $\|\cdot\|$-norm equal to $1$. Therefore, there exists an $\omega$-dependent subsequence of $(n_k)_{k\in \N}$ denoted again $(n_k)_{k\in\N}$ along which $(\Lambda_{n_k})$ converges to an accumulation point $\Lambda_{\acc}$. Lemma~\ref{lem:acc point pos} ensures that $\Lambda_\acc(\rho_0)\neq 0$ $\P^{\rho_0}$-almost surely. Since $\rho_0\ll\hat\rho_0$, $\Lambda_\acc(\hat\rho_0)\neq 0$ $\P^{\rho_0}$-almost surely also. 

Let us fix $\omega\in \Omega$ such that the limit of Eq.~\eqref{eq:convergence words} holds and $\Lambda_\acc(\rho_0)\neq 0$. By definition, the set of such $\omega$ is of full $\P^{\rho_0}$-measure.
Let $\rho_\acc=\normalized{\Lambda_\acc(\rho_0)}$ and $\hat\rho_\acc=\normalized{\Lambda_\acc(\hat \rho_0)}$. Since $\rho_0\ll\hat \rho_0$ and $\tr{\Lambda_\acc(\rho_0)}>0$, $\rho_\acc\ll\hat \rho_\acc$. By continuity of the fidelity it follows that, for any $I\in \cup_{p\in \N}\O^p$,
\begin{equation*}
    \left| F \left( \normalized{\Phi_{I}(\rho_\acc)} , \normalized{\Phi_{I}(\hat\rho_\acc)} \right) - F \left( \rho_\acc , \hat\rho_\acc  \right) \right|  \tr{\Phi_{I} (\rho_\acc)} = 0,
\end{equation*}
where $\normalized{\Phi_{I}(\rho_\acc)}$ and $\normalized{\Phi_{I}(\hat\rho_\acc)}$ are defined arbitrarily when $\tr{\Phi_{I} (\rho_\acc)}=0$.

Consider the positive semi-definite operator \(X\) defined in Assumption \(\prox\). Under Assumption \(\irr\), there must exist $J\in \cup_{k\in\N}\O^k$ such that
\begin{equation}\label{eq:mot_qui_connecte}
    \tr{X\,\Phi_J(\rho_\acc)}>0.
\end{equation}
Otherwise, summing over $\O^k$ would yield
\(\tr{X\,\Phi^k(\rho_\acc)}=0\) for all \(k\), and \cite[Proposition~2.2]{evans1978spectral} prevents this eventuality when \irr holds. Since $\rho_\acc\ll\hat\rho_\acc$,  we also have $\tr{X\,\Phi_J(\hat\rho_\acc)}>0.$ 

Let \((I_{n})_{n\in\N}\) be the sequence from Assumption \(\prox\). Then, for all \(n\in\mathbb{N}\),

\begin{equation*}
    \left| F \left( \normalized{\Phi_{I_{n}}  \circ \Phi_{J} (\rho_\acc)} , \normalized{\Phi_{I_{n}} \circ \Phi_{J} (\hat\rho_\acc)} \right) - F \left( \rho_\acc , \hat\rho_\acc \right) \right| \tr{\Phi_{I_{n}} \circ \Phi_{J} (\rho_\acc ) } = 0.
\end{equation*}
Letting \(n\) tend to infinity, we deduce
\begin{equation*}
    \left| F \left( \normalized{\tr{X \Phi_J(\rho_\acc)} Z } ,\normalized{\tr{X \Phi_J(\hat\rho_\acc)} Z } \right) - F \left( \rho_\acc , \hat\rho_\acc \right) \right|\tr{X \Phi_J (\rho_\acc)}  \tr{Z} = 0.
\end{equation*}
Since \(\tr{X \Phi_J (\rho_\acc)} \tr{Z} > 0\), we conclude that
\begin{align*}
    F \left( \rho_\acc, \hat\rho_\acc \right) &= F \left( \normalized{\tr{X \Phi_J (\rho_\acc)} Z}, \normalized{\tr{X \Phi_J (\hat\rho_\acc)} Z} \right) \\
    &= F \left( \normalized{Z}, \normalized{Z} \right) \\
    &= 1.
\end{align*}
As $\omega$ was chosen arbitrarily in a set of full $\P^{\rho_0}$-measure, it follows that $1$ is $\P^{\rho_0}$-almost surely an accumulation point of $(F(\rho_n,\hat \rho_n))$. Since Theorem~\ref{thm:fidelite_sous_mart} implies that $(F(\rho_n,\hat \rho_n))$ converges $\P^{\rho_0}$-almost surely, it must converge to $1$ $\P^{\rho_0}$-almost surely, and the forward implication is proved.

\medskip

Let us now proceed with the proof of the reverse implication.
Let $\{\hat\rho_0^{(j)}\}_{j=1}^{d^2}\subset\S(\H)$ be a family of
full-rank states spanning $\B(\H)$, and let $\rho_0=\rho_0^{(1)}$. 
Since every $\hat\rho_0^{(j)}$ is full rank, we have $\rho_0 \ll  \hat\rho_0^{(j)}$ for all $j$. 
Therefore, the assumption on the right–hand side of the equivalence applies to every pair $(\rho_0,\hat\rho_0^{(j)})$.
Let us consider again
$$\Lambda_n=\frac{\Phi_{\omega_n}\circ\dotsb\circ\Phi_{\omega_1}}{\|\Phi_{\omega_n}\circ\dotsb\circ\Phi_{\omega_1}\|}$$
and take an accumulation point $\Lambda_\acc.$ 
Since each $\rho_0^{(j)}$ is full-rank, we have $\Lambda_\acc(\rho_0^{(j)}) \neq 0$ for all $j$. 
By convergence of the fidelity and its continuity, $\P^{\rho_0}$–a.s. we have
\[
F\oo\normalized{\Lambda_\acc(\rho_0)},\,\normalized{\Lambda_\acc(\hat\rho_0^{(j)})}\cc = 1,
\qquad j = 1,\dots,d^2.
\]
Thus for each $j$ there exists $\alpha_j>0$ such that
\[
\Lambda_\acc(\hat\rho_0^{(j)})=\alpha_j\,\Lambda_\acc(\rho_0).
\]
Since $\{\hat\rho_0^{(j)}\}$ spans $\B(\H)$ and $\Lambda_\acc$ is linear,
this implies $\range(\Lambda_\acc)=\mathbb C \Lambda_\acc(\rho_0)$ and therefore $\rank\Lambda_\acc=1$. 
Thus, for $\P^{\rho_0}$–almost every sequence $(i_n)\in \Omega$,
the normalized composition
\[
\frac{\Phi_{i_n}\circ\cdots\circ\Phi_{i_1}}{\|\Phi_{i_n}\circ\cdots\circ\Phi_{i_1}\|}
\]
admits a rank–one accumulation point. Hence by setting $I_n=(i_1,\dotsc,i_n)$, \cont holds, and the reverse implication is proved.
\end{proof}

\section{Uniqueness of the invariant measure}\label{sec:uniqueness}
This section is devoted to studying the invariant measure of the quantum trajectory, from which an ergodic theorem will follow. To this end, we make explicit the Markov kernel describing the quantum trajectories. Let us equip the set $\S(\H)$ with its Borel $\sigma$-algebra $\fB$. The Markov kernel $\Pi$ of the quantum trajectory is defined as
$$\Pi(\rho, B)=\sum_{i \in \O} \mathds 1_B \left(\normalized{\Phi_i(\rho)} \right)\tr{\Phi_i(\rho)},$$
for all $\rho \in \S(\H)$ and all $B \in \fB$. Here, $\Pi(\rho, B)$ represents the probability that, starting from the state $\rho$, the next state of the trajectory lies in the subset $B$ of $\S(\H).$ The assumption that $\Phi=\sum_{i\in\O}\Phi_i$ is trace-preserving ensures the conservation of probability, $\Pi(\rho, \S(\H)) = 1$ for any $\rho\in \S(\H)$.

Given a probability measure $\nu$ over $\S(\H)$, we define the action of the Markov kernel $\Pi$ on $\nu$ by
$$
\nu \Pi (A) = \int_{\S(\H)} \Pi(\rho, A)\,\mathrm{d}\nu(\rho),
\qquad A \in\fB.
$$
An invariant measure is a probability measure $\nu_\inv$ on $(\S(\H),\fB)$ that satisfies the equation:
$$\nu_\inv \Pi = \nu_\inv.$$
The first natural question is whether such a measure exists. The Markov--Kakutani theorem guarantees existence whenever the state space is compact and the Markov kernel is Feller. 
Recall that $\Pi$ acts on bounded measurable functions $f$ on $\S(\H)$ by
\[
\Pi f(\rho) \;=\; \int_{\S(\H)} f(\rho')\,\Pi(\rho,\mathrm{d}\rho'),
\]
and that $\Pi$ is Feller if it maps continuous functions to continuous ones, which holds immediately in our setting. Since $\mathcal{S}$ is compact, the assumptions of the theorem are satisfied, and the existence of an invariant probability measure follows. The main result of this section then concerns its uniqueness and the convergence toward it.

To prepare the statement of the theorem, we introduce the following definitions. 
The first one is the Wasserstein distance that we shall use to quantify convergence toward the invariant measure~$\nu_{\mathrm{inv}}$. 
For two probability measures $\nu_1,\nu_2$ on $(\S(\H),\fB)$, the $1$-Wasserstein distance is defined, via Kantorovich–Rubinstein duality, by
$$
W_1(\nu_1,\nu_2)
=\sup_{f\in \operatorname{Lip}_1(\S(\H))}\Bigg|
\int_{\S(\H)} f(\rho)\,\mathrm{d}\nu_1(\rho)
-\int_{\S(\H)} f(\rho)\,\mathrm{d}\nu_2(\rho)
\Bigg|,
$$
where
$$
\operatorname{Lip}_1(\S(\H))
=\big\{f:\S(\H)\to\mathbb{R}\ \text{s.t.}\ |f(\rho)-f(\sigma)|\le \norm{\rho - \sigma}\ \text{for all } \rho,\sigma\in \S(\H)\big\}.
$$

We now recall the notion of the period of an irreducible quantum channel, which is the natural analogue of the classical notion for Markov chains. 
Let $\Phi:\mathcal{B}(\mathcal{H})\to\mathcal{B}(\mathcal{H})$ be an irreducible quantum channel. 
The \emph{period} $\ell$ of $\Phi$ is defined as the largest integer $\ell \ge 1$ for which there exists an orthogonal decomposition
$$
\mathcal{H} = \bigoplus_{j=0}^{\ell-1} \mathcal{H}_j
$$
such that
$$
\Phi\big(\mathcal{B}(\mathcal{H}_j)\big)\subset \mathcal{B}(\mathcal{H}_{j+1 \bmod \ell})
\qquad \text{for all } j=0,\dots,\ell-1.
$$
In this case, $\Phi$ acts as a cyclic permutation on the subspaces $\mathcal{H}_j$. 
When $\ell=1$, the channel is said to be \emph{aperiodic}. The \emph{primitive} quantum channels are those which are irreducible and aperiodic.

Here is the main result of this section.
\begin{theorem}[Uniqueness and convergence to the invariant measure]\label{thm:invariant_measure_convergence}
Assume that $\irr$ and $\cont$ hold. Then the Markov kernel $\Pi$ admits a unique invariant probability measure $\nu_{\mathrm{inv}}$. Moreover, if $\ell$ denotes the period of the channel $\Phi$, then for any initial distribution $\nu$ on $(\S(\H),\fB),$
$$
\limn W_1\!\left(\frac{1}{\ell}\sum_{r=0}^{\ell-1}\nu\Pi^{\ell n+r},\,\nu_{\mathrm{inv}}\right)=0.
$$
\end{theorem}

To prove Theorem \ref{thm:invariant_measure_convergence}, we need some intermediate results, which we gather below. The proof will follow the same structure as that of \cite[Theorem~1.1]{benoist2019invariant} with assumption \prox replacing the purification assumption used there. These two assumptions play the same role, ensuring that the true trajectory and its estimate, constructed solely from the measurement outcomes, become asymptotically indistinguishable. This property is crucial, as it leads both to the uniqueness of the invariant measure and to the convergence of the trajectory toward it.

We now introduce the measurable structure and the extended probability measure \(\P_{\nu}\), which allow us to describe the distribution of the quantum trajectory starting from a random initial state.
We denote
$$\J_n = \fB \otimes \F_n \quad \text{and} \quad \J = \fB \otimes \F,$$
so that $(\S(\H) \times \Omega,\J)$ becomes a measurable space. Let $\nu$ be a probability measure on $(\S(\H),\fB)$. We extend it to a probability measure $\P_{\nu}$ on $(\S(\H) \times \Omega, \J)$ by setting
\begin{equation}\label{eq:def_proba_etendue}
    \P_{\nu}(B \times O) = \int_{B}  \P^{\rho}(O) \mathrm{d}\nu(\rho),
\end{equation}
for all $B \in \fB$ and $O\in \F$. By construction, the restriction of $\mathbb{P}_\nu$ to $\S(\H) \otimes \{\emptyset, \Omega\}$ coincides with $\nu$. The restriction of $\mathbb{P}_\nu$ to $\mathcal F$ also admits a standard expression.

For any probability measure $\nu$ on $\S(\H),$ define
$$\rho_\nu := \mathbb E_\nu [\rho_0] = \int_{\S(\H)}\rho_0 \mathrm d \nu (\rho_0).$$
Since $\S(\H)$ is convex, we indeed have $\rho_\nu\in \S(\H)$.

In particular, for any invariant measure $\nu_\inv,$ 
$$\rho_{\nu_\inv}=\E[\nu_\inv][]{\rho_0}=\E[\nu_\inv][]{\rho_1}=\E[\nu_\inv][]{\Phi(\rho_0)}=\Phi(\rho_{\nu_\inv}),$$ 
where the equalities follow from invariance of $\nu_\inv$ and linearity of $\Phi$. Thus by Assumption \irr, we obtain $\rho_{\nu_\inv} = \rho_\inv.$

\begin{proposition}[\cite{benoist2019invariant}, Proposition~2.1]\label{prop:same measure on ouctomes}
    The marginal of $\P_\nu$ on $\F$ is $\P^{\rho_\nu}$.
\end{proposition}
\begin{proof}
    By definition of $\P_\nu$, for any $O\in \F$, $\P_\nu(\S(\H)\times O)=\E[\nu][]{\P^\rho(O)}$. Since by definition $\rho\mapsto\P^\rho$ is affine, $\P_\nu(\S(\H)\times O)=\P^{\rho_\nu}(O)$ and the proposition is proved.
\end{proof}

In what follows, we will use the estimated trajectory \((\hat\rho_n)\) as a demonstration tool. Recall that it is the sequence of random variables defined, for each $n$, by
\[
\hat\rho_n(\omega)
  = \normalized{\Phi_{\omega_n} \circ \cdots \circ \Phi_{\omega_1}(\hat\rho_0)},
  \qquad \omega \in \Omega.
\]
 From now on, we choose a full-rank estimated initial state $\hat\rho_0 > 0$. In particular, this guarantees that the condition $\rho_0 \ll \hat\rho_0$ is automatically satisfied $\P_\nu$-almost surely, regardless of the distribution $\nu$ of $\rho_0$. This ensures the well-posedness of $(\hat\rho_n).$ 
 \medskip
 
The next corollary is a straightforward extension of Theorem~\ref{thm:asymp_stab_imperfect} to the case of random initial states.

\begin{corollary}\label{cor:convergence_estimate_under_P_nu}
    Suppose that $\irr$ and $\prox$ hold, and that $\hat\rho_0>0.$ Then, for any probability measure $\nu$ on $(\S(\H),\fB)$,
    \[
        \lim_{n\to\infty} F(\rho_n,\hat\rho_n) = 1 \quad \P_\nu\text{-a.s.}
    \]
\end{corollary}
\begin{proof} Theorem \ref{thm:asymp_stab_imperfect} implies that, for any $\rho\in\S(\H)$, $\P^\rho(\lim_n F(\rho_n,\hat \rho_n)=1)=1.$ Moreover, by definition $\delta_\rho\otimes\P^{\rho}(\cdot)=\P_{\delta_\rho}(\cdot)$. Hence, for any $\rho\in \S(\H)$, $\P_{\delta_\rho}(\lim_n F(\rho_n,\hat \rho_n)=1)=1$. Therefore, for a general initial distribution $\nu$ on $(\S(\H),\fB),$ we obtain
\begin{equation*}\P_\nu(\lim_n F(\rho_n,\hat \rho_n)=1)=\E[\nu][]{\P_{\delta_\rho}(\lim_n F(\rho_n,\hat \rho_n)=1)}=1.\qedhere\end{equation*}
\end{proof}

The following result is a reformulation of the Perron–Frobenius theorem of Evans and H\o egh-Krohn \cite[Theorems~4.2~and~4.4]{evans1978spectral}, in the form given for instance in \cite[Theorem~3.3]{benoist2019invariant}.
\begin{theorem}\label{thm:convergence_iterates_quantum_channel}
Assume that \irr holds and let $\ell$ be the period of the channel $\Phi$. Then, for all $\rho \in \S(\H)$,
$$
\limn {\norm*[][]{\frac{1}{\ell}\sum_{r=0}^{\ell-1}\Phi^{\ell n+r}(\rho)\;-\;\rho_\inv}}=0.
$$
\end{theorem}
We recall the following lemma from \cite{benoist2024exponentially} which will be useful in the proof of the convergence toward the invariant measure.
\begin{lemma}[Lemma~4.1 in \cite{benoist2024exponentially}]\label{lem:rho_mapsto_Prho_Lipschitz}
    The map $\rho \mapsto \P^\rho$ is Lipschitz. More precisely, for all states $\rho,\sigma \in \S(\H),$
    $$\left\|\mathbb{P}^\rho-\mathbb{P}^\sigma\right\|_{TV}:=\sup _{A \in \F}\left|\mathbb{P}^\rho(A)-\mathbb{P}^\sigma(A)\right| \leq  \norm[][]{\rho-\sigma}.$$
\end{lemma}

We are now in a position to prove Theorem \ref{thm:invariant_measure_convergence}. We only prove the convergence as it implies the uniqueness of the invariant measure.
\begin{proof}[Proof of Theorem \ref{thm:invariant_measure_convergence}]
    Let $\nu_\inv$ be a $\Pi$-invariant probability measure.
    
    The variational definition of the Wasserstein distance is unaffected by a global shift of the functions $f$ by a constant. We can therefore restrict ourselves to $1$-Lipschitz functions that are equal to $0$ at some arbitrary $\sigma\in \S(\H)$. Since $\sup_{\rho,\sigma\in\S(\H)}\|\rho-\sigma\|=2$, for any of these functions their supremum norm $\|f\|_\infty$ is bounded by $2$. Thus, let $f \in \operatorname{Lip}_1(\S(\H))$ be arbitrary such that $\|f\|_\infty\leq 2$. 
    
    For convenience, we write, for all $n,$ $\bar \Pi_n := \frac1\ell \sum_{r=0}^{\ell-1} \Pi^{\ell n+r}.$
    For any $p,q$ such that $n = p + q,$
\begin{align}\label{eq:epsilon sur 3}
\begin{split}
|\E[\nu \bar \Pi_n][]{f(\rho_0)} -& \E[\nu_\inv][]{f(\rho_0)}|\\
=& \abs{\E[\nu \bar \Pi_p][]{f(\rho_q)} - \E[\nu_\inv][]{f(\rho_q)}} \\
\leq &\abs{\E[\nu \bar \Pi_p]{f(\rho_q)} - \E[\nu \bar \Pi_p]{f(\hat\rho_q)}} + \abs{\E[\nu \bar \Pi_p]{f(\hat\rho_q)} - \E[\nu_\inv]{f(\hat\rho_q)}}  \\
&\quad + \abs{\E[\nu_\inv]{f(\hat\rho_q)} - \E[\nu_\inv][]{f(\rho_q)}} \\
\leq& \E[\nu \bar \Pi_p]{\norm[]{\rho_q - \hat\rho_q}}  + \abs{\E[\nu \bar \Pi_p]{f(\hat\rho_q)} - \E[\nu_\inv]{f(\hat\rho_q)}}  + \E[\nu_\inv]{\norm[]{\rho_q - \hat\rho_q}}.
\end{split}
\end{align}
The first equality is a consequence of the Markov property. The first inequality is a consequence of triangle inequality. For the last inequality we used the fact that $f$ is $1$-Lipschitz. 

Let $\varepsilon >0.$ From Corollary \ref{cor:convergence_estimate_under_P_nu}, which holds for any 
probability measure $\nu$ over $\S(\H)$, we have $\P_\nu$-a.s. $F(\rho_n,\hat\rho_n)\xrightarrow[n\to\infty]{} 1$. Hence, by \cite[Proposition~5]{fuchs2002cryptographic},
$$\norm[]{\rho_n-\hat\rho_n}\xrightarrow[n\to\infty]{\P_\nu\mbox{-a.s.}} 0.$$
Using this convergence for $\nu=\nu\bar\Pi_p$ in the first term and for
$\nu=\nu_{\mathrm{inv}}$ in the third one, Lebesgue’s dominated convergence theorem implies that for any $p$ fixed and $q$ large enough,
\begin{equation}\label{eq:conv_estimates}
\E[\nu \bar \Pi_p]{\norm[]{\rho_q - \hat\rho_q}} < \varepsilon/3 \quad\mbox{and}\quad  \E[\nu_\inv]{\norm[]{\rho_q - \hat\rho_q}} <\varepsilon/3.
\end{equation}
For the second term, the $\F$-measurability of 
$\hat\rho_q$, Proposition~\ref{prop:same measure on ouctomes} and the identities 
$\overline{\rho}_p:=\rho_{\nu \bar \Pi_p} = \frac{1}{\ell}\sum_{r=0}^{\ell-1} \Phi^{\ell p+r}(\rho_\nu),$ and 
$\rho_{\nu_\inv}=\rho_\inv$ imply
$$\abs{\E[\nu \bar \Pi_p]{f(\hat\rho_q)} - \E[\nu_\inv]{f(\hat\rho_q)}}=\abs{\E[][\overline{\rho}_p]{f(\hat \rho_q)}-\E[][\rho_\inv]{f(\hat\rho_q)}}.$$
Then, the inequality $\abs{\E[][\overline{\rho}_p]{f(\hat \rho_q)}-\E[][\rho_\inv]{f(\hat\rho_q)}}\leq \|f\|_\infty\|\P^{\overline{\rho}_q}-\P^{\rho_\inv}\|_{TV}$ and Lemma~\ref{lem:rho_mapsto_Prho_Lipschitz} lead to
$$\abs{\E[\nu \bar \Pi_p]{f(\hat\rho_q)} - \E[\nu_\inv]{f(\hat\rho_q)}} \leq 2 \norm*[]{\frac1 \ell \sum_{r=0}^{\ell-1} \Phi^{\ell p+r}(\rho_\nu) - \rho_\inv}.$$
Remark that this bound is uniform in $q$. Then using Theorem \ref{thm:convergence_iterates_quantum_channel}, for $p$ large enough,
\begin{equation*}\label{eq:conv TV}
\sup_{q\in\N}\abs{\E[\nu \bar \Pi_p]{f(\hat\rho_q)} - \E[\nu_\inv]{f(\hat\rho_q)}} < \varepsilon/3.
\end{equation*}
Fix $p\in \N$ such that this bound holds. Then, for any $n\in \N$ large enough Equation~\ref{eq:conv_estimates} holds for $q=n-p$. It follows from Equation~\ref{eq:epsilon sur 3} that, for any large enough $n\in \N$,
$$W_1\left(\frac{1}{\ell}\sum_{r=0}^{\ell-1}\nu\Pi^{\ell n+r},\,\nu_{\mathrm{inv}}\right) < \varepsilon,$$
which concludes the proof.
\end{proof}

We finish this article by stating a consequence of the last theorem, a strong law of large numbers for arbitrary continuous functions. Its proof is a straightforward adaptation of the one of \cite[Theorem~3.1]{benoist2023limit}, and is therefore omitted.
\begin{theorem}[Convergence of the ergodic mean]\label{thm:convergence_of_the_ergodic_mean} 
    Let $g$ be a continuous function on $\S(\H).$ Then, for any probability measure $\nu$ on $\S(\H),$
    $$\frac{g(\rho_1)+\dotsb+g(\rho_n)}{n} \xrightarrow[n\to\infty]{\P_\nu\mbox{-a.s.}}\mathbb E_{\nu_{\inv}}[g].$$
\end{theorem}

\bigskip

\noindent\textbf{Acknowledgements.} The research of TB was partly funded by ANR project DYNACQUS, grant number ANR-24-CE40-5714. This work was supported by the ANR projects Q-COAST (ANR-19-CE48-0003) and
IGNITION (ANR-21-CE47-0015).

\noindent\textbf{Conflicts of interest.} The authors declare no conflict of interest with respect to the present article.

\bibliographystyle{alpha} 
\bibliography{refs} 

@article{sayrin2011real,
  title={Real-time quantum feedback prepares and stabilizes photon number states},
  author={Sayrin, Cl{\'e}ment and Dotsenko, Igor and Zhou, Xingxing and Peaudecerf, Bruno and Rybarczyk, Th{\'e}o and Gleyzes, S{\'e}bastien and Rouchon, Pierre and Mirrahimi, Mazyar and Amini, Hadis and Brune, Michel and others},
  journal={Nature},
  volume={477},
  number={7362},
  pages={73--77},
  year={2011},
  publisher={Nature Publishing Group UK London}
}

@article{ocone1996asymptotic,
  title={Asymptotic stability of the optimal filter with respect to its initial condition},
  author={Ocone, Daniel and Pardoux, Etienne},
  journal={SIAM Journal on Control and Optimization},
  volume={34},
  number={1},
  pages={226--243},
  year={1996},
  publisher={SIAM}
}

@article{van2009observability,
  title={Observability and nonlinear filtering},
  author={Van Handel, Ramon},
  journal={Probability theory and related fields},
  volume={145},
  number={1},
  pages={35--74},
  year={2009},
  publisher={Springer}
}

@techreport{delyon1988lyapunov,
  title = {Lyapunov Exponents for Filtering Problems},
  url = {http://dx.doi.org/10.21236/ADA459564},
  DOI = {10.21236/ada459564},
  institution = {Defense Technical Information Center},
  author = {Delyon,  Bernard and Zeitouni,  Ofer},
  year = {1988},
  month = apr 
}

@article{amini2014stability,
  title={Stability of continuous-time quantum filters with measurement imperfections},
  author={Amini, Hadis and Pellegrini, Cl{\'e}ment and Rouchon, Pierre},
  journal={Russian Journal of Mathematical Physics},
  volume={21},
  number={3},
  pages={297--315},
  year={2014},
  publisher={Springer}
}

@article{benoist2024exponentially,
  title={Exponentially fast selection of sectors for quantum trajectories beyond non demolition measurements},
  author={Benoist, Tristan and Greggio, Linda and Pellegrini, Cl{\'e}ment},
  journal={Annales Henri Poincaré},
  year={2025}
}

@article{kummerer2004pathwise,
  title={A pathwise ergodic theorem for quantum trajectories},
  author={K{\"u}mmerer, Burkhard and Maassen, Hans},
  journal={Journal of Physics A: Mathematical and General},
  volume={37},
  number={49},
  pages={11889},
  year={2004},
  publisher={IOP Publishing}
}

@article{furstenberg1963noncommuting,
  author       = {Furstenberg, Harry},
  title        = {Non commuting random products},
  journal      = {Transactions of the American Mathematical Society},
  volume       = {108},
  number       = {3},
  pages        = {377--428},
  year         = {1963},
  publisher    = {American Mathematical Society},
  doi          = {10.2307/1993601}
}

@article{guivarch1980proprietes,
  author       = {Guivarc'h, Yves and Raugi, Albert},
  title        = {Propri{\'e}t{\'e}s de contraction d'un semi-groupe de matrices inversibles et applications},
  journal      = {Israel Journal of Mathematics},
  volume       = {65},
  number       = {2},
  pages        = {165--196},
  year         = {1989},
  publisher    = {Springer},
  doi          = {10.1007/BF02784794}
}

@inproceedings{somaraju2012design,
  title={Design and stability of discrete-time quantum filters with measurement imperfections},
  author={Somaraju, Abhinav and Dotsenko, Igor and Sayrin, Clement and Rouchon, Pierre},
  booktitle={2012 American Control Conference (ACC)},
  pages={5084--5089},
  year={2012},
  organization={IEEE}
}

@article{benoist2025quantum_a,
  title={Quantum trajectories. Spectral gap, quasi-compactness \& limit theorems},
  author={Benoist, Tristan and Hautecoeur, Arnaud and Pellegrini, Cl{\'e}ment},
  journal={Journal of Functional Analysis},
  volume={289},
  number={5},
  pages={110932},
  year={2025},
  publisher={Elsevier}
}

@article{maassen2006purification,
  title={Purification of quantum trajectories},
  author={Maassen, Hans and K{\"u}mmerer, Burkhard},
  journal={Lecture Notes-Monograph Series},
  pages={252--261},
  year={2006},
  publisher={JSTOR}
}

@article{jozsa1994fidelity,
  title={Fidelity for mixed quantum states},
  author={Jozsa, Richard},
  journal={Journal of modern optics},
  volume={41},
  number={12},
  pages={2315--2323},
  year={1994},
  publisher={Taylor \& Francis}
}

@article{amini2021asymptotic,
  title={On asymptotic stability of quantum trajectories and their Cesaro mean},
  author={Amini, Nina H and Bompais, Ma{\"e}l and Pellegrini, Cl{\'e}ment},
  journal={Journal of Physics A: Mathematical and Theoretical},
  volume={54},
  number={38},
  pages={385304},
  year={2021},
  publisher={IOP Publishing}
}

@article{benoist2019invariant,
  title={Invariant measure for quantum trajectories},
  author={Benoist, Tristan and Fraas, Martin and Pautrat, Yan and Pellegrini, Cl{\'e}ment},
  journal={Probability Theory and Related Fields},
  volume={174},
  pages={307--334},
  year={2019},
  publisher={Springer}
}

@inproceedings{benoist2021invariant,
  title={Invariant measure for stochastic Schr{\"o}dinger equations},
  author={Benoist, Tristan and Fraas, Martin and Pautrat, Yan and Pellegrini, Cl{\'e}ment},
  booktitle={Annales Henri Poincar{\'e}},
  volume={22},
  pages={347--374},
  year={2021},
  organization={Springer}
}

@inproceedings{bompais2023asymptotic,
  title={Asymptotic stability of non-demolition quantum trajectories with measurement imperfections},
  author={Bompais, Ma{\"e}l and Amini, Nina},
  booktitle={2023 62nd IEEE Conference on Decision and Control (CDC)},
  pages={5926--5932},
  year={2023},
  organization={IEEE}
}

@article{benoist2023limit,
  title={Limit theorems for quantum trajectories},
  author={Benoist, Tristan and Fatras, Jan-Luka and Pellegrini, Cl{\'e}ment},
  journal={Stochastic Processes and their Applications},
  volume={164},
  pages={288--310},
  year={2023},
  publisher={Elsevier}
}

@article{rouchon2011fidelity,
  title={Fidelity is a sub-martingale for discrete-time quantum filters},
  author={Rouchon, Pierre},
  journal={IEEE Transactions on automatic control},
  volume={56},
  number={11},
  pages={2743--2747},
  year={2011},
  publisher={IEEE}
}

@article{handel2009stability,
  title={The stability of quantum Markov filters},
  author={HANDEL, RAMON VAN},
  journal={Infinite Dimensional Analysis, Quantum Probability and Related Topics},
  volume={12},
  number={01},
  pages={153--172},
  year={2009},
  publisher={World Scientific}
}

@article{evans1978spectral,
  title={Spectral Properties of Positive Maps on C*-Algebras},
  author={Evans, David E and H{\o}egh-Krohn, Raphael},
  journal={Journal of the London Mathematical Society},
  volume={2},
  number={2},
  pages={345--355},
  year={1978},
  publisher={Wiley Online Library}
}

@article{bougerol,
  title = {Products of {{Random Matrices}} with Applications to {{Schr{\"o}dinger Operators}}. 1985},
  author = {Bougerol, P. and Lacroix, J.},
  year = 1985,
  journal = {Birgh{\"a}user, Boston-Basel-Stuttgart}
}

@article{fuchs2002cryptographic,
  title={Cryptographic distinguishability measures for quantum-mechanical states},
  author={Fuchs, Christopher A and Van De Graaf, Jeroen},
  journal={IEEE Transactions on Information Theory},
  volume={45},
  number={4},
  pages={1216--1227},
  year={2002},
  publisher={IEEE}
}
\end{document}